\pdfoutput=1
\documentclass{siamart190516}

\usepackage[capitalize]{cleveref}

\usepackage{microtype}%

\usepackage{array}
\usepackage{cite} 
\usepackage{mathtools}                
\usepackage{amsfonts}

\theoremstyle{plain}
\newtheorem{claim}[theorem]{Claim}
\newtheorem{observation}[theorem]{Observation}

\theoremstyle{nonumberplain}
\newtheorem{remark}{Remark}

\DeclareMathOperator{\polylog}{\mathrm{polylog}}

\DeclareMathOperator{\area}{\mathit{area}}

\let\eps\varepsilon
\newcommand{\sinr}{\mbox{\it sinr}\,}
\newcommand{\tsinr}{\mbox{\it s}\widetilde{\mbox{\it in}}\mbox{\it r}\,}
\newcommand{\intrf}{\mbox{\it intrf}\,}
\newcommand{\tintrf}{\mbox{\it i}\widetilde{\mbox{\it ntr}}\mbox{\it f}\,}
\newcommand{\nrg}{\mbox{\it nrg}\,}
\newcommand{\tnrg}{\widetilde{\mbox{\it nrg}}\,}
\newcommand{\dist}{\mbox{\it dist}\,}
\newcommand{\pa}{\mbox{\it PA}\,}
\newcommand{\ps}{\mbox{\it PS}\,}

\def\RR{\mathbb{R}}

\definecolor{darkblue}{rgb}{0,0,0.6}
\def\boris#1{\textcolor{darkblue}{\textsc{Boris says}: \textsf{#1}}}

\title{Resolving SINR Queries in a Dynamic Setting\thanks{An earlier version of this paper (excluding \cref{sec:extensions} and \cref{sec:ultimate} and some of the proofs) was presented at ICALP'18~\cite{AronovBK18}. Work on this paper was initiated at the Fifth Workshop on Geometry and Graphs,
Bellairs Research Institute, Barbados, 2017. \funding{Boris Aronov and Matthew Katz were supported by a joint grant 2014/170 from the US-Israel Binational Science Foundation.}}}
\author{%
  Boris Aronov\thanks{%
    Department of Computer Science and Engineering,
    Tandon School of Engineering, New York University, Brooklyn, NY~11201,
    USA; \email{boris.aronov@nyu.edu}.
    \funding{Work on this paper by Boris Aronov was supported
      by NSF Grants CCF-11-17336, CCF-12-18791, and CCF-15-40656.%
    }}
  \and
  Gali Bar-On\thanks{%
    Department of Computer Science, Ben-Gurion University of the Negev,
    Beer-Sheva 84105, Israel;
    \email{galibar@post.bgu.ac.il}.}
  \and
  Matthew J. Katz\thanks{%
    Department of Computer Science, Ben-Gurion University of the Negev,
    Beer-Sheva 84105, Israel;
    \email{matya@cs.bgu.ac.il}.
    \funding{Work on this paper by Matthew Katz was supported
      by grant 1884/16 from the Israel Science Foundation.}}}

\begin{document}
\maketitle

\headers{Resolving SINR Queries in a Dynamic Setting}{Aronov, Bar-On, and Katz}

\begin{abstract}
  We consider a set of transmitters broadcasting simultaneously on the same frequency under the SINR model.  Transmission power may vary from one transmitter to another, and a transmitter's signal strength at a given point
  is modeled by the transmitter's power divided by some constant power $\alpha$ of the distance it traveled.  Roughly, a receiver at a given location can hear a specific transmitter only if the transmitter's signal is stronger by a specified ratio than the signals of all other transmitters combined.  An SINR query is to determine whether a receiver at a given location can hear any transmitter, and if yes, which one.

  An \emph{approximate} answer to an SINR query is such that one gets a definite \textsc{yes} or definite \textsc{no}, when the ratio between the strongest signal and all other signals combined is well above or well below the reception threshold, while the answer in the intermediate range is allowed to be either \textsc{yes} or \textsc{no}.

  We describe compact data structures that support approximate SINR queries in the plane in a dynamic context, i.e., where transmitters may be inserted and deleted over time. We distinguish between two main variants --- uniform power and non-uniform power. In both variants the preprocessing time is $O(n\,\polylog n)$ and the amortized update time is $O(\polylog n)$, while the query time is $O(\polylog n)$ for uniform power, and randomized time $O(\sqrt{n}\,\polylog n)$ with high probability for non-uniform power.	
	
	Finally, we observe that in the static context the latter data structure can be implemented differently, so that the query time is also $O(\polylog n)$, thus significantly improving all previous results for this problem.
\end{abstract}

\section{Introduction}
\label{sec:intro}

The \emph{Signal to Interference plus Noise Ratio (SINR) model} attempts to predict whether a wireless transmission is received successfully, in a setting consisting of multiple simultaneous transmitters in the presence of background noise.
Let $S=\{s_1, \ldots, s_n\}$ be a set of $n$ transmitters (distinct points in the plane), and let $p_i$ denote the transmission power of $s_i$, for $i=1,\ldots,n$. Let $q$ be a receiver (a point in the plane). According to the SINR model, $q$ \emph{receives} $s_i$ if and only if
\[
  \sinr(q,s_i) \coloneqq \frac{\frac{p_i}{|qs_i|^\alpha}}{\sum_{j \ne i}\frac{p_j}{|qs_j|^\alpha}+N} \ge \beta\,,
\]
where $\alpha \ge 1$ and $\beta > 1$ are constants, $N$ is a constant representing the background noise, and $|ab|$ is the Euclidean distance between points $a$ and $b$.

Observe that, since $\beta > 1$, $q$ may receive at most one transmitter ---  the one ``closest'' to it, namely, the one for which the value $\frac{p_i}{|qs_i|^\alpha}$ is maximum, or, equivalently, %
$\frac{1}{p_i^{1/\alpha}}|qs_i|$ is minimum.
Thus, one can partition the plane into $n$ not necessarily connected reception regions~$R_i$, one per transmitter in $S$, plus an additional region~$R_\emptyset$ consisting of all points where none of the transmitters is received. This partition is called the \emph{SINR diagram} of~$S$ \cite{aeklpr-sdciawn-12}.

In their seminal paper, Avin et al.~\cite{aeklpr-sdciawn-12} studied properties of SINR diagrams, focusing on the \emph{uniform power} version where $p_1=p_2=\cdots=p_n$. Their main result is that in this version the reception regions~$R_i$ are convex and fat. In the \emph{non-uniform power} version, on the other hand, the reception regions are not necessarily connected, and their connected components are not necessarily convex or fat.  In fact, they may contain holes~\cite{klpp-twn-11}.

An \emph{SINR query} is: Given a receiver $q$, find the sole transmitter $s$ that may be received by $q$ and determine whether it is indeed received by $q$, i.e., whether or not $\sinr(q,s) \ge \beta$.
A~natural question is: How quickly can one answer an SINR query, following a preprocessing stage in which data structures of total size nearly linear in $n$ are constructed? However, it~seems unlikely that the answer is significantly sub-linear (as the degree of the polynomials describing region boundaries is high), so the research has focused on preprocessing to facilitate efficient \emph{approximate} SINR queries.

\begin{table}
  \begin{center}

    \extrarowheight 2pt
    \renewcommand{\arraystretch}{1.6}
	\begin{tabular}{| l |l | l | l | }
		    \hline
	    \textbf{Power} & \textbf{Preprocessing} & \textbf{Space} & \textbf{Query} \\ \hline
	Uniform \cite{aeklpr-sdciawn-12} &	$O(\dfrac{n^2}{\eps})$ ($O(\dfrac{n}{\eps^{2.5}}\log^4 n \log\log n)$ \cite{AK-TAlg}) & $O(\dfrac{n}{\eps})$ & $O(\log n)$  \\*[1.1ex] \hline
	Non-Uniform \cite{klpp-twn-11} & $O(\dfrac{\varphi'}{\eps^2} n^2)$ & $O(\dfrac{\varphi'n}{\eps^2})$ & $O(\dfrac{\varphi}{\eps}\log n )$ \\*[1.3ex] \hline

  	\end{tabular}
		\end{center}
  \caption{Approximate SINR queries in a static setting --- previous results; $\varphi$ is an upper bound on the fatness parameters of the reception regions and $\varphi' \ge \varphi^2$.}
	\label{table:previous_results}
\end{table}

The approach of such research has been to construct a data structure which approximates the underlying SINR diagram, and use it for answering approximate SINR queries, by performing point-location queries in this structure. That is, given a query point $q$, first find the sole candidate $s_i$ that may be received at $q$ (say, by searching in the appropriate Voronoi diagram), and then perform a point-location query to approximately determine whether $q$ is in $R_i$ or not.
Two different notions of approximation have been used. In the first~\cite{aeklpr-sdciawn-12}, it~is guaranteed that the uncertain answer is only given infrequently, namely, the area of the uncertain region associated with $R_i$ is at most $\eps \cdot \area (R_i)$, for a prespecified parameter $\eps>0$. In the second~\cite{klpp-twn-11}, it is guaranteed that for every point in the uncertain region the SIN ratio is within an $\eps$-neighborhood of $\beta$.
See \cref{table:previous_results} for a summary of previous results; see also \cite{KantorLPP15} for related work. In addition, Aronov and Katz~\cite{AK-TAlg} obtained several results for \emph{batched} approximate SINR queries, using the latter notion of approximation; for example, one can perform $n$ simultaneous approximate queries in a network with $n$ transmitters at polylogarithmic amortized cost per query.

Given $\eps > 0$,\footnote{For simplicity of presentation, we will assume hereafter that $n>1/\eps$.}
 an \emph{approximate SINR query} is: Given a receiver $q$, find the sole transmitter~$s$ that may be received by $q$ and return a value $\tsinr(q,s)$, such that $(1 - \eps)\sinr(q,s) \le \tsinr(q,s) \le (1 + \eps)\sinr(q,s)$. Thus, unless $(1-\eps)\beta \le \tsinr(q,s) < (1+\eps)\beta$, the value $\tsinr(q,s)$ enables us to determine definitely whether or not $s$ is received by $q$.

In this paper, we devise efficient algorithms for handling dynamic approximate SINR queries. That is, given $S$, $\alpha, \beta$, and $N$, as above, and $\eps > 0$, we describe algorithms for answering approximate SINR queries after some initial preprocessing, in a setting where transmitters may be added to or deleted from $S$. We analyze our algorithms by the usual measures, namely, data structure size and preprocessing, query, and update times.

To the best of our knowledge, these are the first data structures to support dynamic approximate SINR queries.  In contrast with previous work on approximate SINR queries, our algorithms do not compute an approximation of the underlying SINR diagram.
We distinguish between two main variants of the problem --- the uniform power version and the non-uniform one. The preprocessing time in both cases is $O(n\,\polylog n)$, while the query and update time is $O(\polylog n)$ for the uniform version, and $O(\sqrt{n}\,\polylog n)$ for the non-uniform version. Thus, our solution for the \emph{dynamic} uniform version is comparable to the best known solutions for the \emph{static} uniform version. For the non-uniform version, our solution is the first one with bounds that depend only on $n$ and $\eps$ and not on other parameters of the input, both in the static and dynamic settings.

Moreover, for the non-uniform version in the static setting, we present the first algorithm for handling approximate SINR queries in $O(\polylog n)$ time, after $O(n\,\polylog n)$-time preprocessing. The algorithm is similar to its dynamic counterpart, however, the demanding stages of the latter algorithm can be implemented more efficiently in the static setting (with obvious changes to the data structure).

In addition to the obvious motivation for devising algorithms for dynamic approximate SINR queries, we mention another important application of our results. \emph{Successive Interference Cancellation (SIC)} is a technique that enables (in some circumstances) a~receiver~$q$ to receive a specific transmitter~$t$, even if $t$ cannot be received at $q$ in SINR sense.  Informally, our results support SIC; if $t$'s signal is the $k$th strongest at  $q$, then, through a sequence of $O(k)$ queries and updates, we can determine whether $q$ can decode $t$'s signal from the combined signal.
In contrast, Avin et al.~\cite{AvinCHKLPP17} construct a uniform-power static data structure of size $O(\eps^{-1}n^{10})$ which enables one to determine in $O(\log n)$ time whether $t$ can be received by $q$ using SIC. Their result is not directly comparable to ours, however, they guarantee logarithmic query regardless of the number of transmitters that need to be canceled before $t$ can be heard, and their approximation model is quite different from ours. See remark in \cref{sec:sic}, in which we argue that in practice $k$ does not exceed $O(\log n)$.

\paragraph*{Our results and organization}

In \cref{sec:uniform}, we consider the uniform power variant of the problem, that is we assume that all the transmitters have the same transmission power. We describe a data structure of size $O(n\,\polylog n)$ that supports approximate SINR queries in $O(\polylog n)$ time and updates in $O(\polylog n)$ amortized time.\footnote{We are ignoring here the dependency on the approximation factor $\eps$.} In \cref{sec:extensions}, we obtain a similar result for the common case where the ratio between the maximum power and the minimum power is bounded by a constant, or the number of distinct transmission powers is bounded by a constant. In \cref{sec:non-uniform}, we consider the non-uniform power variant, i.e., we assume arbitrary-power transmitters. For this variant, we describe a data structure of size $O(n\,\polylog n)$ that can answer approximate SINR queries in randomized time $O(\sqrt{n}\,\polylog n)$ with high probability and can perform updates in amortized $O(\polylog n)$ time. Our dynamic data structures can be used to determine whether a receiver $q$ can receive a specific transmitter $t$ through successive interference cancellation, which is the topic of \cref{sec:sic}. Finally, in \cref{sec:ultimate}, we consider the non-uniform power variant in a static setting. We observe that the costly stages in the query algorithm in a dynamic setting can be implemented more efficiently (with obvious changes to the data structure) when the set of transmitters is fixed. We thus obtain a data structure of expected size $O(n\,\polylog n)$ supporting approximate SINR queries in expected $O(\polylog n)$ time; see \cref{table:previous_results} for the previous bounds for this problem. We note that the construction time of all our data structures is $O(n\,\polylog n)$.

\section{Uniform power}
\label{sec:uniform}

We first discuss the slightly easier case of uniform power.
Let $q$ be a receiver and let $s$ be the closest transmitter to $q$. Set $\intrf(q)=\sum_{s' \in S\setminus\{s\}} \frac{1}{|q s'|^\alpha}$, then $\sinr(q,s)=\frac {\frac{1}{|qs|^\alpha}} {\intrf(q)}$.\footnote{For clarity of presentation, we assume hereafter that there is no noise, i.e., $N=0$. Our algorithms extend to the situation where noise is present in a straightforward manner.}
When $s$ is the transmitter closest to~$q$, we will simply write $\sinr(q)$ instead of $\sinr(q,s)$.
Fix $\eps > 0$.  We wish to compute a value $\tsinr(q)$ satisfying $(1-\eps)\sinr(q) \leq \tsinr(q) \leq \sinr(q)$. We show below how to compute a value $\tintrf(q)$ such that $\intrf(q) \leq \tintrf(q) \leq (1+\eps)\intrf(q)$ and then simply set $\tsinr(q)=\frac{\frac{1}{|qs|^\alpha}}{\tintrf(q)}$. Clearly, we have

\begin{claim}
\label{eps_and_delta}
Under the above assumption,
$(1-\eps)\sinr(q) < \tsinr(q) \leq \sinr(q)$.
\end{claim}

We start with a slower but easier to describe solution and then refine it.

\subsection{Annuli}
\label{annuli}
Let $\eps > 0$.
Let $q$ be a receiver and let $s \in S$ be the closest transmitter to $q$. Let $s_1,\ldots,s_{n-1}$ be the transmitters in $S \setminus \{s\}$, and assume without loss of generality that $s_1$ is the second closest transmitter to $q$, among all the transmitters in $S$, breaking ties arbitrarily. Recall that $\intrf(q) = \sum_{i=1}^{n-1} \frac{1}{|q s_i|^\alpha}$ and that we wish to compute a value $\tintrf(q)$ such that $\intrf(q) \le \tintrf(q) \le (1+\eps)\intrf(q)$.

We will need the following simple observation.
\begin{observation}
\label{obs:intrf}
$\intrf(q)$ is the sum of $n-1$ positive terms of which $\frac{1}{|qs_1|^\alpha}$ is the largest, so
$\frac{1}{|qs_1|^\alpha} \le \intrf(q) \le \frac{n-1}{|qs_1|^\alpha} < \frac{n}{|qs_1|^\alpha}$.
\end{observation}

\subsubsection{Query algorithm}
\label{sec:alg_desc}

Let $q$ be a query point. First, we find $s$ and $s_1$, the closest and the second closest transmitters to $q$, respectively. Next, we divide the transmitters in $S \setminus \{s\}$ into two subsets, $S_c$ and $S_f$, where $S_c$ consists of all transmitters that are `close' to $q$ and $S_f$ consists of all transmitters that are `far' from $q$. More precisely, set $r = (\frac{2n}{\eps})^{1/\alpha} \cdot |qs_1|$, then $S_c$ consists of all the transmitters in $S \setminus \{s\}$ whose distance from $q$ is less than $r$, and $S_f$ consists of all the remaining transmitters.
We now approximate the contribution of each of these subsets to $\intrf(q)$.

The contribution of a transmitter $s_i$ in $S_f$ to the sum $\intrf(q)$ is
\[
\frac{1}{|qs_i|^\alpha} \le \frac{1}{r^\alpha} = \frac{\eps}{2n|qs_1|^\alpha}\,,
\]
and the combined contribution of the transmitters in $S_f$ is at most $|S_f| \cdot \frac{\eps}{2n|qs_1|^\alpha} \le \frac{\eps}{2|qs_1|^\alpha}$.

\begin{figure}[h]
    \centering
        \includegraphics[width=0.4\textwidth]{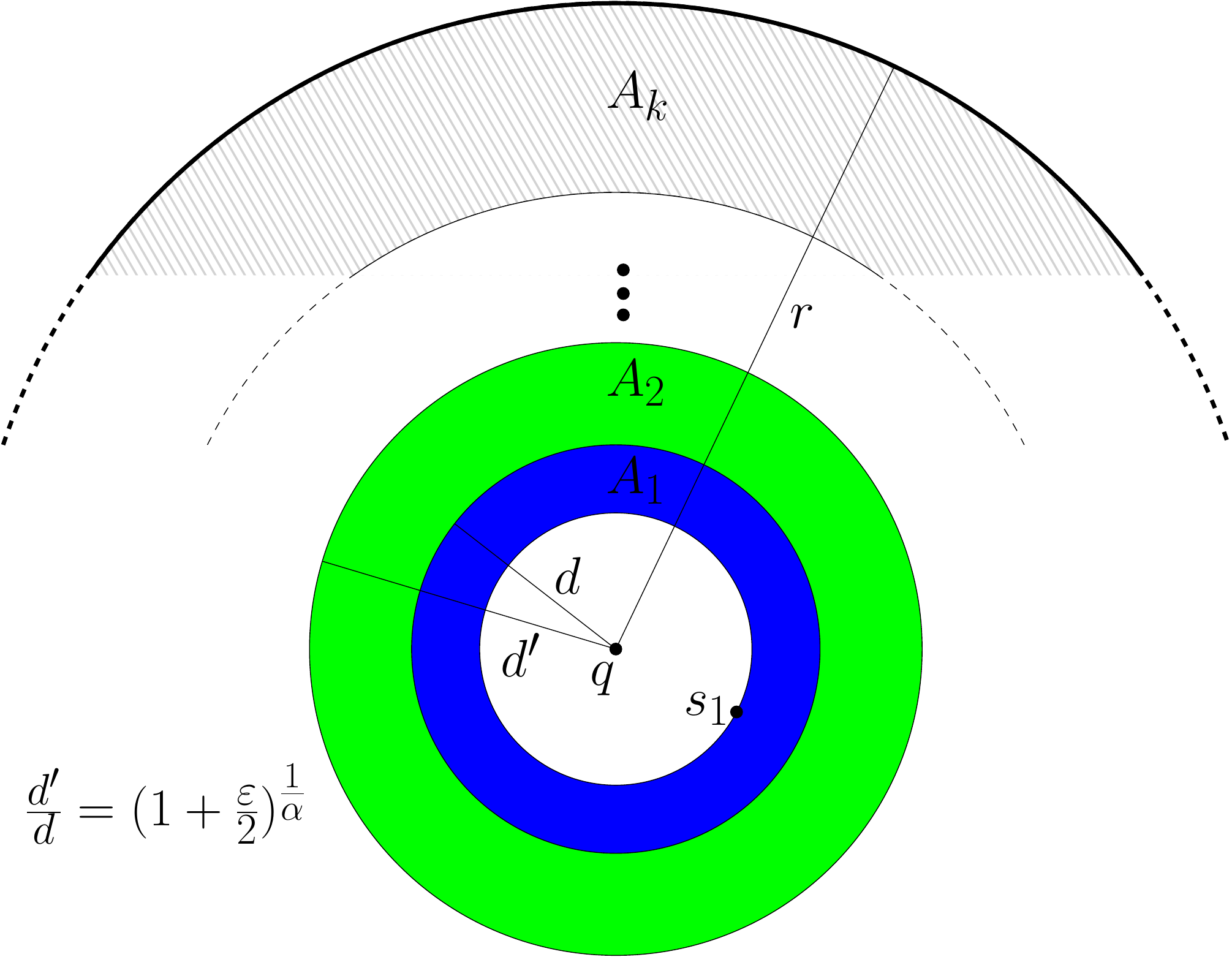}.
    \caption{Partitioning $A_q(|qs_1|,r)$ into annuli.}
    \label{annuli_partition}
\end{figure}

We denote the annulus centered at $q$ with inner radius $r_1$ and outer radius $r_2$ by $A_q(r_1,r_2)$.
In order to approximate the overall contribution of the transmitters in $S_c$, we partition the annulus $A_q(|qs_1|,r)$ into $k$ semi-open annuli, $A_1,\ldots,A_k$, such that the ratio of the outer to the inner radius of $A_j$ is $(1+\frac{\eps}{2})^\frac{1}{\alpha}$ (except for $A_k$ whose corresponding ratio is at most $(1+\frac{\eps}{2})^\frac{1}{\alpha}$); see \cref{annuli_partition}. By semi-open we mean that the inner circle of $A_j$ is contained in $A_j$, but the outer circle is not. Now, for each $A_j$, we approximate the contribution of each transmitter $s_i \in S_c \cap A_j$ by $\frac{1}{d^\alpha}$, where $d$ is the inner radius of $A_j$, that is, we approximate the contribution of $s_i$ by moving it to the inner circle of $A_j$. We prove below (\cref{cor:approx_S_c_contribution}) that this yields a $(1 + \frac{\eps}{2})$-approximation of the overall contribution of the transmitters in $S_c$ to $\intrf(q)$.

\begin{lemma}
\label{lem:moving_s_i}
Let $s_i \in S_c$ and let $A=A_q(d,d')$ be the annulus to which $s_i$ belongs.
Then, by moving $s_i$ to the inner circle of $A$, one obtains a $(1+\frac{\eps}{2})$-approximation of the contribution of $s_i$ to $\intrf(q)$.
\end {lemma}
\begin{proof}
Since $s_i \in A$, $d \le |qs_i| < d'$. Moreover, by construction, $d'/d \le (1+\frac{\eps}{2})^\frac{1}{\alpha}$. So, the ratio of our approximation to the real contribution of $s_i$ is
\[
  \frac{1/d^\alpha}{1/|qs_i|^\alpha} = \frac{|qs_i|^\alpha}{d^\alpha} < (\frac{d'}{d})^\alpha \le 1+\frac{\eps}{2}\,.
\]
\end{proof}
\begin{corollary}
\label{cor:approx_S_c_contribution}
By doing this for each transmitter in $S_c$, one obtains a $(1 + \frac{\eps}{2})$-approximation of the overall contribution of the transmitters in $S_c$ to $\intrf(q)$.
\end{corollary}

It remains to show that $\tintrf(q)$, which is the sum of the approximations for $S_f$ and for $S_c$, satisfies the requirements, i.e., that $\intrf(q) \le \tintrf(q) \le (1 + \eps)\intrf(q)$.
From the description above it is clear that $\tintrf(q) \ge \intrf(q)$, so we only need to show that $\tintrf(q) \le (1 + \eps)\intrf(q)$. Indeed,
\[
\tintrf(q) \le \frac{\eps}{2|qs_1|^\alpha} + (1 + \frac{\eps}{2}) \sum_{s_i \in S_C} \frac{1}{|qs_i|^\alpha} \le
\frac{\eps}{2}\intrf(q) + (1 + \frac{\eps}{2})\intrf(q) = (1 + \eps)\intrf(q)\,,
\]
where the second inequality is based on \cref{obs:intrf}.

\subsubsection{Implementation}
\label{sec:impl_circular}
We first show that $k$, the number of annuli into which the annulus $A_q(|qs_1|,r)$ is partitioned, is small.

\begin{lemma}\label{lem:number_annuli}
  $k = O(\frac{1}{\eps}\log n)$.
\end{lemma}
\begin{proof}
Clearly, $k = \lceil \log_{(1+\frac{\eps}{2})^\frac{1}{\alpha}} \frac{r}{|qs_1|} \rceil$. But,
\[
\log_{(1+\frac{\eps}{2})^\frac{1}{\alpha}} \frac{r}{|qs_1|} =
\log_{(1+\frac{\eps}{2})^\frac{1}{\alpha}} (\frac{2n}{\eps})^\frac{1}{\alpha} =
\frac{\log (\frac{2n}{\eps})^\frac{1}{\alpha} }{\log (1+\frac{\eps}{2})^\frac{1}{\alpha}} =
\frac{\log \frac{2n}{\eps}}{\log (1+\frac{\eps}{2})}\,,
\]
and since $2^x \le 1+x$, for $0 \le x \le 1$, we obtain
\[
  \log_{(1+\frac{\eps}{2})^\frac{1}{\alpha}} \frac{r}{|qs_1|} \le \frac{\log \frac{2n}{\eps}}{\frac{\eps}{2}} = O(\frac{1}{\eps}\log \frac{n}{\eps})\,.
\]
Now, since we are assuming that $n>1/\eps$, $\log \frac{n}{\eps} < 2 \log n = O(\log n)$ and $k=O(\frac{1}{\eps}\log n)$.
\end{proof}

In the preprocessing stage we compute the following data structures for the set of transmitters $S$.

\subparagraph*{Dynamic nearest neighbor}

A data structure due to Chan~\cite{Chan19} can be used for dynamic 2D nearest-neighbor queries.
A set of points in the plane can be maintained dynamically in a linear-size data structure, so as to support insertions, deletions, and nearest-neighbor queries. Each insertion takes $O(\log^2 n)$ amortized deterministic time, each deletion takes $O(\log^4 n)$ amortized deterministic time, and each query takes $O(\log^2 n)$ worst-case deterministic time, where $n$ is the size of the set of points at the time the operation is performed; see 
  also the data structure of Kaplan \textsl{et al.}~\cite{DBLP:conf/soda/KaplanMRSS17plus} with slightly worse performance.

\subparagraph*{Dynamic disk range counting}
We start with the construction of Matou\v{s}ek \cite{m-ept-92}: In linear space and $O(n\log n)$ time one can preprocess a set of $n$ points in $\mathbb{R}^d$ to support semi-group halfspace range queries in $O(n^{1-1/d}\polylog n)$ time. A point can be deleted in $O(\log n)$ amortized time and inserted in $O(\log^2n)$ amortized time.  Lifting circles to points in $\mathbb{R}^3$ in the standard manner, we obtain a linear-space, $O(n \log n)$ time, $O(n^{2/3}\polylog n)$ disk range counting query, $O(\log n)$ amortized delete, $O(\log^2 n)$ amortized insert data structure.  We do not attempt to optimize this ingredient, as we replace this infrastructure with a more efficient one in the following section.

\bigskip

Given a query point $q$, we find $s$ and $s_1$ (the closest and second closest transmitters) using the data structure for dynamic nearest neighbor; both the data structure of Chan~\cite{Chan19} and Kaplan \textsl{et al.}~\cite{DBLP:conf/soda/KaplanMRSS17plus} can be modified to return both the first and second nearest neighbors~\cite{ChanPC,MulzerPC}.
Next, we compute the distance $r$ and partition the annulus $A_q(|qs_1|,r)$ into $k$ annuli, as described above.
Now, we calculate $\tintrf(q)$ as follows. We first compute the size of the set $S_f$ by performing a disk counting query with the circle of radius $r$ centered at $q$ and subtracting the answer from $n-1$; we initialize $\tintrf(q)$ to $|S_f| \cdot \frac{\eps}{2n|qs_1|^\alpha}$.
Next, for each of the $k$ annuli, we compute the number $x$ of points of $S$ lying in it, as the difference in the numbers of points in the two disks defined by its bounding circles, obtained by counting queries.  We then increment $\tintrf(q)$ by $\frac{x}{d^\alpha}$, where $d$ is the radius of the inner circle of the current annulus.

An update is performed by updating the two underlying data structures.

We omit the detailed performance analysis of this version, as a better data structure is described next.

\subsection{Polygonal rings}
\label{sec:poly}

We now present a more efficient solution, which is similar to the previous one, except that we replace the circular annuli by polygonal rings. Set $x=(1+\frac{\eps}{2})^\frac{1}{\alpha}$, and consider any three circles $C_0,C_1,C_2$ centered at $q$, such that $r_1/r_0 = r_2/r_1 = \sqrt x$, where $r_i$ is the radius of $C_i$. %
Set $l = \bigl\lceil\frac{\pi}{\sqrt{2- \frac{2}{\sqrt{x}}}}\bigr\rceil$, and
let $B_i$ be the regular $l$-gon inscribed in $C_i$, so that one of its vertices lies on the upward vertical ray through $q$, for $i=1,2$. We now show that $C_{i-1}$ is contained in $B_i$, for $i=1,2$, and therefore, the polygonal ring defined by $B_1$ and $B_2$ is contained in the annulus $A_q(r_0,r_2)$.  

\begin{figure}[h]
    \centering
        \includegraphics[width=0.32\textwidth]{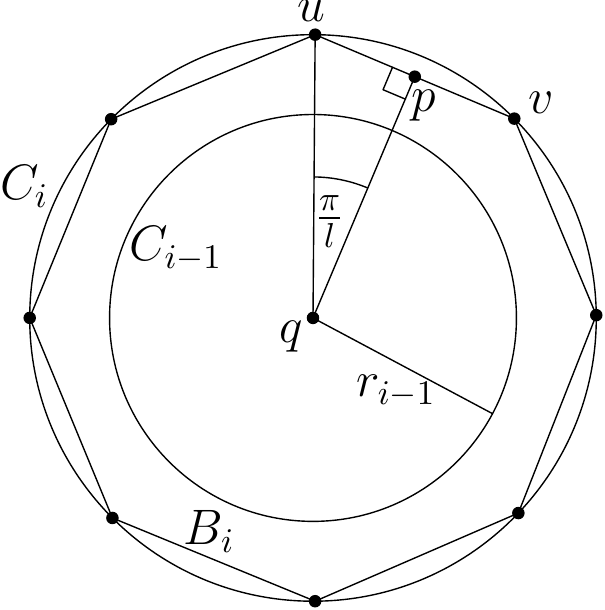}.
    \caption{Proof of \cref{cl:poly_is_included}. The circles $C_{i-1}$ and $C_i$ and the regular $l$-gon $B_i$.}
    \label{fig:poly3}
\end{figure}

We will need the following inequality: for $y\in[0,1]$ and $z \in (0,1]$,
\begin{equation}
  (1+y)^z \ge 1+(2^z-1)y.
  \label{eq:inequality}
\end{equation}
\begin{proof}
The assertion holds with equality when $y$ is $0$ or $1$. The right-hand side is a linear function of $y$ and the left-hand side is a concave function of $y$.
\end{proof}

\begin{claim}
\label{cl:number_edges}
$l=O(1/\sqrt{\eps})$.
\end{claim}
\begin{proof}
Recall that $l = \bigl\lceil \frac{\pi}{\sqrt{2 - \frac{2}{\sqrt{x}}}} \bigr\rceil$ and $x=(1+\frac{\eps}{2})^\frac{1}{\alpha}$.  
Using \eqref{eq:inequality}, we deduce that $\sqrt x = (1+\eps/2)^{1/(2\alpha)} \ge 1 + (2^{1/(2\alpha)}-1)\eps/2$ and
\[
  \frac{1}{1-\frac{1}{\sqrt{x}}}=\frac{\sqrt{x}}{\sqrt{x}-1} \le \frac{1 + (2^{1/(2\alpha)}-1)\eps/2}{(2^{1/(2\alpha)}-1)\eps/2}=O(1/\eps).
\]
We conclude that $l=\bigl\lceil\frac{\pi}{\sqrt{2 - \frac{2}{\sqrt{x}}}} \bigr\rceil=O(1/\sqrt{\eps})$.
\end{proof}

\begin{claim}
\label{cl:poly_is_included}
$C_{i-1}$ is contained in $B_i$, for $i=1,2$.
\end{claim}
\begin{proof}
We need to prove that $|qp| \ge r_{i-1}$, where $p$ is the midpoint of an edge $uv$ of $B_i$, see \cref{fig:poly3}.
Clearly, $\angle qpu = \frac{\pi}{2}$ and $\angle uqp = \frac {\pi}{l}$. Therefore,
\begin{equation}
\label{eq:qp}
|qp| = r_i \cos \frac{\pi}{l} = \sqrt{x} r_{i-1} \cos \frac{\pi}{l}\,.
\end{equation}
Now, recall that $l \ge \frac{\pi}{\sqrt{2 - \frac{2}{\sqrt{x}}}}$.
Hence, $\frac{\pi}{l} \le \sqrt{2- \frac{2}{\sqrt{x}}}$, or $(\frac{\pi}{l})^2 \le 2- \frac{2}{\sqrt{x}}$.
Rearranging and diving by 2, we get that $\frac{1}{\sqrt{x}} \le 1-\frac{(\frac{\pi}{l})^2}{2}$.
But, by the Taylor series for $\cos \theta$, we know that $\cos \theta \ge 1 - \frac{\theta^2}{2}$, so
$\cos \frac{\pi}{l} \ge 1 - \frac{(\frac{\pi}{l})^2}{2} \ge \frac{1}{\sqrt{x}}$.
Thus, together with \cref{eq:qp} we obtain that $|qp| \ge r_{i-1}$, as claimed.
\end{proof}

\begin{corollary}
\label{poly_is_included}
The polygonal ring defined by $B_1$ and $B_2$ is contained in an annulus centered at $q$ with radii ratio $x=(1+\frac{\eps}{2})^\frac{1}{\alpha}$.
\end{corollary}

\begin{figure}[h]
    \centering
        \includegraphics[width=0.5\textwidth]{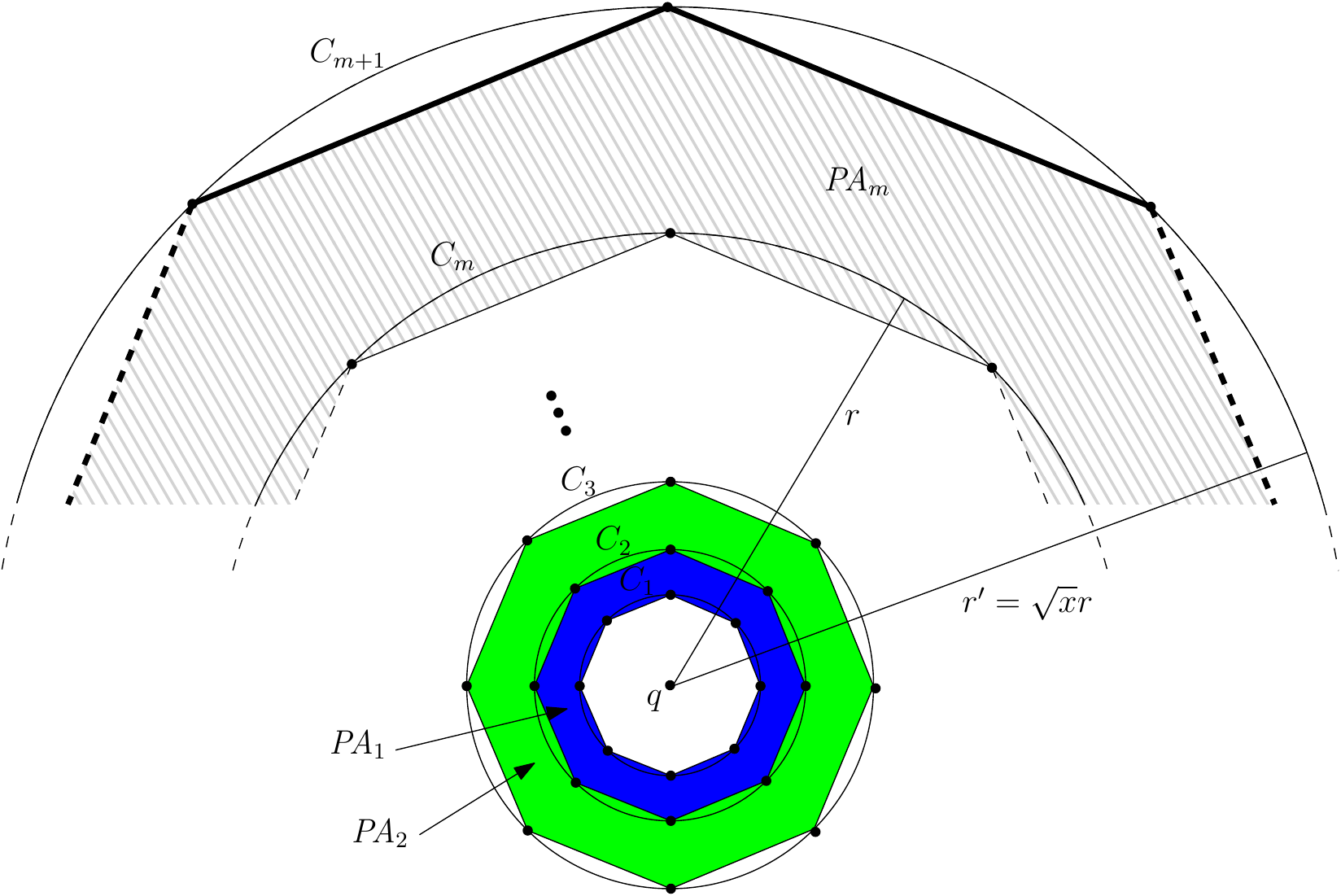}.
    \caption{The polygonal rings $\pa_1,\ldots,\pa_m$. The ring $\pa_j$ lies between the circles $C_{j-1}$ and $C_{j+1}$, for $j > 1$.}
    \label{fig:poly_annuli}
\end{figure}

\subsubsection{Query algorithm}
\label{sec:poly_alg_desc}
We highlight the differences with the query algorithm from \cref{sec:alg_desc}.
Recall that we divided the transmitters into two subsets according to whether they were closer or farther than $r$ from the query point $q$.
We adjust the definitions slightly by setting $r' = |qs_1|x^{m/2}$, where $m$ is the smallest integer for which $|qs_1|x^{m/2} \ge \sqrt{x}r$, and considering a transmitter close to $q$ whenever it lies in the interior of the regular $l$-gon inscribed in the circle of radius $r'$ centered at $q$, see \cref{fig:poly_annuli}. The set of such transmitters is the new $S_c$; the remaining transmitters constitute $S_f$.
The contribution of $s_i\in S_f$ to the sum $\intrf(q)$ is, by \cref{cl:poly_is_included},
\[
\frac{1}{|qs_i|^\alpha} \le \frac{1}{(r')^\alpha} \le \frac{1}{r^\alpha} = \frac{\eps}{2n|qs_1|^\alpha}\,.
\]
Thus the overall contribution of the transmitters in $S_f$ is again at most $\frac{\eps}{2|qs_1|^\alpha}$.

We now partition $A_q(|qs_1|,r')$ into $m$ annuli, each with outer-to-inner radius ratio~$\sqrt{x}$.
For each of the $m+1$ circles defining these annuli, draw the regular $l$-gon inscribed in it.
Let $\pa_1,\ldots,\pa_m$ be the resulting sequence of polygonal rings, numbered from the innermost outwards,  see \cref{fig:poly_annuli}; each ring is semi-open: it includes its inner, but not its outer boundary. By  \cref{cl:poly_is_included}, each $\pa_j$, $j>1$, is contained in the union of two consecutive annuli, which in turn is an annulus of ratio~$x$; $S_c \cap \pa_1$ is contained in the innermost annulus. Also notice that $m=O(k)$, where $k$ is the number of annuli in the circular annulus version, so, by \cref{lem:number_annuli}, %
$m = O(\frac{1}{\eps}\log n)$.

For each ring $\pa_j$, we bound from above the contribution of each $s_i \in S_c \cap \pa_j$ by $1/d^\alpha$, where $d$ is the inner radius of the annulus of ratio $x$ containing $\pa_j$, if $j>1$, and $d=|qs_1|$, if $j=1$. By \cref{lem:moving_s_i} and the subsequent corollary, we obtain a $(1 + \frac{\eps}{2})$-approximation of the overall contribution of the transmitters in $S_c$ to $\intrf(q)$; $\tintrf(q)$ is obtained by combining the two estimates, one from $S_c$ and one from $S_f$.

\begin{figure}[h]
    \centering
        \includegraphics[width=0.3\textwidth]{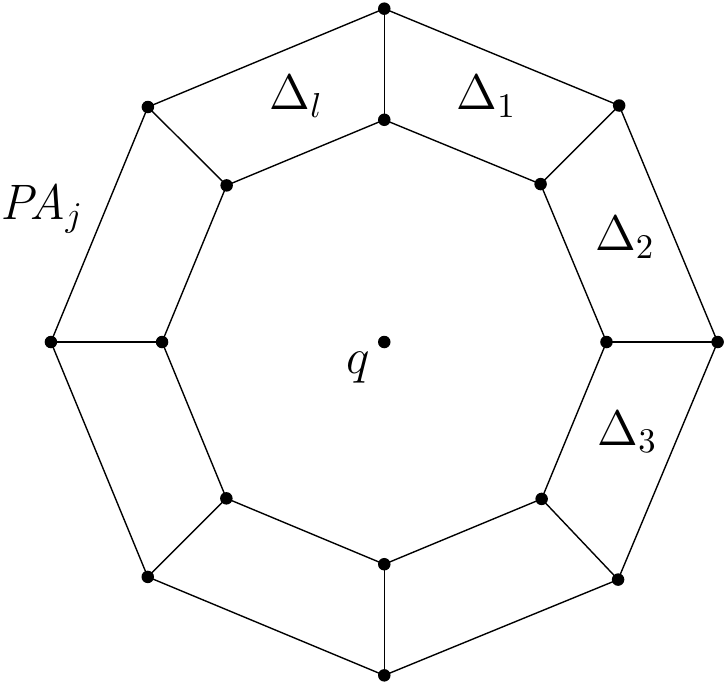}.
    \caption{$\pa_j$ is the union of $l$ trapezoids.}
    \label{fig:trapezoids}
\end{figure}

\subsubsection{Implementation}
Each polygonal ring $\pa_j$ is the union of $l$ isosceles trapezoids; moreover the $i$th trapezoids of all rings are homothets of each other (refer to \cref{fig:trapezoids}) and therefore are delimited by lines of exactly three different orientations.
In the preprocessing stage we compute the following data structures for the set of transmitters $S$.
\subparagraph*{Dynamic nearest neighbor}
  The data structure of Chan~\cite{Chan19} or Kaplan \emph{et al.}~\cite{DBLP:conf/soda/KaplanMRSS17plus}  (see \cref{sec:impl_circular}).
\subparagraph*{Dynamic trapezoid range counting}
  We use $l$ instances of the data structure, one for each family of trapezoids. For the $i$th family we build a three-level orthogonal range counting structure, one for each of the three edge orientations of the trapezoids in the family. The answer to a trapezoid range counting query is the number of points of $S$ lying in the trapezoid.

  A standard three-level orthogonal range counting structure requires $O(n\log^2n)$ space, is constructed in $O(n\log^2n)$ time, and supports $O(\log^3n)$-time range queries \cite{bcko-cgaa-08}.  It can be modified to support insertions and deletions in $O(\log^3n)$ amortized time using the standard partial-rebuilding technique~\cite{o-ddds-83,a-rs-04}.  (One can use any of several different optimized variants of these structures \cite{wl-arrcd-85,c-fadsi-88}. For example, He and Munro~\cite{he-munro-14} describe one with linear space and $O((\log n/\log \log n)^2)$ worst-case query and amortized update time; we stay with comparison-based algorithms and do not attempt to optimize the polylogarithmic factors.)

\medskip

Now, given a query point $q$, we find its closest and second closest transmitter using the data structure for dynamic nearest neighbor in $O(\log^2 n)$ time,
compute the distance $r'$, and construct the (polygonal) rings $\pa_1,\ldots,\pa_m$, where %
$m = O(\frac{1}{\eps}\log n)$.
For each ring $\pa_j$ we proceed as follows. For each of the $l$ trapezoids $\Delta_i$ forming $\pa_j$, we perform an orthogonal range counting query in the $i$th data structure. Let $n_j$ be the sum of the $l$ results. Unless $j = 1$, we add to the value being computed the term $n_j \frac{1}{r_{j-1}^\alpha}$, where $r_{j-1}$ is the radius of the inner circle $C_{j-1}$ of the annulus containing $\pa_j$. If $j=1$, we simply add the term $n_1 \frac{1}{r_1^\alpha}$. Finally, we add to the value being computed the term $|S_f| \cdot \frac{\eps}{2n|qs_1|^\alpha} = (n-1 - \sum_{j=1}^m n_j) \cdot \frac{\eps}{2n|qs_1|^\alpha}$.

In summary, to implement an SINR query, we need to perform one search for the nearest and second-nearest neighbor, followed by $O(\frac{l}{\eps}\log n)=O(\frac{1}{\eps^{3/2}}\log n)$ range searches.

An update is applied to all the underlying data structures.
The following theorem summarizes the main result of this section.

\begin{theorem}\label{th:uniform_polygonal}
  Given the locations of $n$ uniform-power transmitters, one can preprocess them in $O((n/\sqrt{\eps})\log^2 n)$ time and space into a data structure that can answer approximate SINR queries in $O((1/\eps^{3/2})\log^4n)$ time.  Transmitters can be inserted in $O((1/\sqrt{\eps})\log^3n)$ and deleted in $O((1/\sqrt{\eps})\log^3n)$ amortized time.
\end{theorem}

\section{Extensions}
\label{sec:extensions}

In this section, we extend our result for the uniform power setting to the more general setting where the ratio of the maximum to minimum power is bounded by a constant. That is, let $p_{\min}$ and $p_{\max}$ be the transmission powers of the weakest and strongest transmitters, respectively, and assume, without loss of generality, that $p_{\min} = 1$.  Below we show that the bounds obtained for the uniform power setting (i.e., \cref{th:uniform_polygonal}) continue to hold in the more general setting where $\frac{p_{\max}}{p_{\min}} = p_{\max}$ is bounded by a constant.

We first consider an easy case, where the number of distinct transmission powers is at most a constant. In this case, we do not make any assumption on the ratio $\frac{p_{\max}}{p_{\min}}$. Then, we consider the case mentioned above, where the number of different transmission powers is unlimited, but the ratio $\frac{p_{\max}}{p_{\min}}$ is bounded by a constant.

\subsection{A few powers}\label{unq_ext}

Let $m$ be the number of distinct transmission powers, $T_1,\ldots,T_m$ be the partition of $S$ into subsets by transmission power, put $t_i=|T_i|$, and let $p_i$ be the transmission power of the transmitters in $T_i$.
For each subset $T_i$, we construct the data structure for uniform-power approximate SINR queries, as in \cref{th:uniform_polygonal}.

Now, let $q$ be a receiver. We denote the transmitter of $T_i$ that is closest to $q$ by $s_i$, for $i = 1,\ldots,m$. Moreover, assume without loss of generality that $\frac{p_1}{|qs_1|^\alpha} \ge \frac{p_i}{|qs_i|^\alpha}$, for $i = 1,\ldots,m$. Finally, set
$\intrf_i(q) = \sum_{s \in T_i \setminus \{s_i\}} \frac{1}{|q s|^\alpha}$ and $\intrf_{T_i}(q)=p_i \cdot \intrf_i(q)$. Then, 
$\intrf(q) = \intrf_{T_1}(q) + \cdots + \intrf_{T_m}(q) + \frac{p_2}{|qs_2|^\alpha} + \cdots + \frac{p_m}{|q s_m|^\alpha}$, and our goal is to compute a value $\tintrf(q)$ such that $\intrf(q) \le \tintrf(q) \le (1+\eps)\intrf(q)$.

\subsubsection{Query algorithm}\label{sec:ext_2_alg_desc}

Perform a uniform-power approximate SINR query with $q$ in each of the $m$ data structures.
Let $\tintrf_i(q)$ be the value that is computed by the data structure for $T_i$ and set $\tintrf_{T_i}(q) = p_i \cdot \tintrf_i(q)$, for $i=1,\ldots,m$. Then, $\intrf_{T_i}(q) \le \tintrf_{T_i}(q) \le (1+\eps)\intrf_{T_i}(q)$. Finally, set $\tintrf(q) = \tintrf_{T_1}(q) + \cdots + \tintrf_{T_m}(q) + \frac{p_2}{|qs_2|^\alpha} + \cdots + \frac{p_m}{|q s_m|^\alpha}$. 
It is easy to see that $\intrf(q) \le \tintrf(q) \le (1+\eps)\intrf(q)$.

\begin{theorem}\label{th:few_powers}
  \label{th:polygonal_few}
  Given the locations of $n$ transmitters with $m$ distict transmitting powers, one can preprocess them in $O((n/\sqrt{\eps})\log^2 n)$ time and space into a data structure that can answer approximate SINR queries in $O((m/\eps^{3/2})\log^4n)$ time.  Transmitters can be inserted in $O((1/\sqrt{\eps})\log^3n)$ and deleted in $O((1/\sqrt{\eps})\log^3n)$ amortized time.
\end{theorem}

\subsection{A small range of powers}\label{ext_bnd}

Assume without loss of generality that $p_{\min}=1$ and $p_{\max}=c$, where $c$ is a constant. We extend our result for the uniform power setting to this more general setting. We first require that $1 < c \le 1+\frac{\eps}{4}$. Then, we apply \cref{th:few_powers} to remove the dependency of $c$ on $\eps$, so that $c$ can be any constant greater than~1.

In this section we assume that $\eps$ is sufficiently small, say, less than $\min\{1, 4(\beta - 1)\}$. Also,
for the sake of clarity, we use circular annuli (rather than polygonal rings), but it should be clear that one can replace the annuli by rings as in the uniform power setting.

\subsubsection{Range depends on \texorpdfstring{$\eps$}{epsilon}}

Let $q$ be a receiver, and consider the closest and second closest transmitters to $q$, by Euclidean distance. Denote the one whose signal strength at $q$ is stronger by $s$ and the other one by $s_1$, and denote all the other transmitters by $s_2,\ldots,s_{n-1}$. Notice that it is still possible that there exists a third transmitter $s_i \in S$, whose signal strength at $q$ is greater than that of $s$. However, in this case, it is easy to see that $\sinr(q,s_i) \ll \beta$ (and of course also $\sinr(q,s) \ll \beta$), so our algorithm will return the correct answer in this case. Let $\intrf(q) = \sum_{1}^{n-1} \frac{p_i}{|q s_i|^\alpha}$. We wish to compute a value $\tintrf(q)$ such that $\intrf(q) \le \tintrf(q) \le (1+\eps)\intrf(q)$.

\paragraph*{Query algorithm}
The query algorithm is almost identical to the circular-annuli query algorithm, so we only list the differences.	
\begin{itemize}
\item
Recall that in the annuli solution the transmitters are divided into two subsets, the close and far subsets, by $r = (\frac{2n}{\eps})^{1/\alpha} \cdot |qs_1|$. Here we set $r=(\frac{2cn}{\eps})^{1/\alpha} \cdot |qs_1|$ to define the two subsets $S_c$ and $S_f$.
Let $s_i \in S_f$, then its contribution to $\intrf(q)$ is
\begin{align*}
  \frac{p_i}{|qs_i|^\alpha} & \le \frac{p_i}{r^\alpha} \le \frac{c}{r^\alpha} = \frac{c}{((\frac{2cn}{\eps})^{1/\alpha} \cdot |qs_1|)^\alpha} \\
  & = \frac{\eps}{2n|qs_1|^\alpha} \le \frac{\eps\cdot p_1}{2n|qs_1|^\alpha} \le \frac{\eps}{2n}\intrf(q)\,.
\end{align*}
This implies that the overall contribution to $\intrf(q)$ of the transmitters in $S_f$ is bounded by $\frac{\eps}{2}\intrf(q)$.
\item
In order to approximate the overall contribution of the transmitters in $S_c$, we partition the annulus $A_q(|qs_1|,r)$ into $k$ semi-open annuli, $A_1,\ldots,A_k$, such that the ratio between their outer radius and inner radius is $(\frac{1}{c}\cdot(1+\frac{\eps}{2}))^\frac{1}{\alpha}$. For each $A_j$, we approximate the contribution of each transmitter $s_i \in S_c \cap A_j$ by $\frac{c}{d^\alpha}$, where $d$ is the inner radius of $A_j$, that is, we approximate the contribution of $s_i$ by moving it to the inner circle of $A_j$ and assigning it power $c$. We prove below (\cref{cor:ext_approx_S_c_contribution}) that this yields a $(1 + \frac{\eps}{2})$-approximation of the overall contribution to $\intrf(q)$ of the transmitters in $S_c$. 
\item
The number of annuli in this version is still $O(\frac{1}{\eps}\log \frac{n}{\eps})$ as proved in \cref{lem:number_annuli_ext} below.
\end{itemize}

\begin{lemma}
Let $s_i \in S_c$ and let $A=A_q(d,d')$ be the annulus to which $s_i$ belongs. 
Then, by moving $s_i$ to the inner circle of $A$ and assigning it power $c$, one obtains a $(1+\frac{\eps}{2})$-approximation of the contribution of $s_i$ to $\intrf(q)$. 
\end {lemma} 
\begin{proof}
Since $s_i \in A$, we know that $d \le |qs_i| < d'$. Recall that $d'/d \le (\frac{1}{c}\cdot(1+\frac{\eps}{2}))^\frac{1}{\alpha}$. Hence, $\frac{|qs_i|}{d} < \frac {d'}{d} \le (\frac{1}{c}\cdot(1+\frac{\eps}{2}))^\frac{1}{\alpha}$, or 
$(\frac{|qs_i|}{d})^\alpha < \frac{1}{c}\cdot (1+\frac{\eps}{2})$.
By rearranging the latter inequality, we obtain $ \frac{c}{d^\alpha}< (1+\frac{\eps}{2})\cdot \frac{1}{|qs_i|^\alpha}$.
Therefore,
\[
\frac{p_i}{|qs_i|^\alpha} \le
\frac{p_i}{d^\alpha} \le
\frac{c}{d^\alpha} <
(1+\frac{\eps}{2})\cdot \frac{1}{|qs_i|^\alpha} \le
(1+\frac{\eps}{2})\cdot \frac{p_i}{|qs_i|^\alpha}.
\]
\end{proof}
\begin{corollary}
\label{cor:ext_approx_S_c_contribution}
By doing this for each transmitter in $S_c$, one obtains a $(1 + \frac{\eps}{2})$-approximation of the overall contribution of the transmitters in $S_c$ to $\intrf(q)$.  
\end{corollary}

\begin{lemma}\label{lem:number_annuli_ext}
The number of annuli $k$ considered by the algorithm is $O(\frac{1}{\eps}\log \frac{n}{\eps})$. 
\end{lemma}
\begin{proof}
Clearly, $k = \lceil \log_{(\frac{1}{c}\cdot(1+\frac{\eps}{2}))^\frac{1}{\alpha}} \frac{r}{|qs_1|} \rceil$. But,
\[
\log_{(\frac{1}{c}\cdot(1+\frac{\eps}{2}))^\frac{1}{\alpha}} \frac{r}{|qs_1|} = 
\log_{(\frac{1}{c}\cdot(1+\frac{\eps}{2}))^\frac{1}{\alpha}} (\frac{2cn}{\eps})^\frac{1}{\alpha} =
\frac{\log (\frac{2cn}{\eps})^\frac{1}{\alpha} }{\log (\frac{1}{c}\cdot(1+\frac{\eps}{2}))^\frac{1}{\alpha}} =
\frac{\log \frac{2cn}{\eps}}{\log (\frac{1}{c}\cdot(1+\frac{\eps}{2}))}\,,
\]
 and, since $1 < c < 1+\frac{\eps}{4}$, it holds that $\frac{1}{c}\cdot(1+\frac{\eps}{2}) > 1+\frac{\eps}{4+\eps} \ge 1+\frac{\eps}{5}$. We conclude that 
\[
\log_{(\frac{1}{c}\cdot(1+\frac{\eps}{2}))^\frac{1}{\alpha}} \frac{r}{|qs_1|} = O(\frac{1}{\eps}\log \frac{n}{\eps})\,.
\] 
\end{proof}

\begin{lemma}
  \label{lem:small_range_eps}
  Fix $\eps>0$.  Given the locations of $n>1/\eps$ transmitters, whose powers differ by a factor of at most $1+\eps/4$, one can preprocess them in $O((n/\sqrt{\eps})\log^2 n)$ time and space into a data structure that can answer approximate SINR queries in $O((1/\eps^{3/2})\log^4n)$ time.  Transmitters can be inserted in $O((1/\sqrt{\eps})\log^3n)$ and deleted in $O((1/\sqrt{\eps})\log^3n)$ amortized time.
\end{lemma}

\subsubsection{Range does not depend on \texorpdfstring{$\eps$}{epsilon}}
\label{ext_comb}
Given a constant $c > 1$ which does not depend on $\eps$, we divide the range of powers into $m = \lceil \log_{1 + \frac{\eps}{4}} c \rceil$ subranges, such that in each subrange, the ratio between the maximum power and minimum power is at most $1 + \frac{\eps}{4}$. We then apply \cref{th:few_powers} and \cref{lem:small_range_eps} to obtain the following theorem. We omit the easy details.

\begin{theorem}
  \label{th:small_range}
	Fix $\eps>0$.  Given the locations of $n>1/\eps$ transmitters, whose powers differ by a factor of at most $c$,
  one can preprocess them in $O((n/\sqrt{\eps})\log^2 n)$ time and space into a data structure that can answer approximate SINR queries in $O((m/\eps^{3/2})\log^4n)$ time, where $m = \lceil \log_{1 + \frac{\eps}{4}} c \rceil$.  Transmitters can be inserted in $O((1/\sqrt{\eps})\log^3n)$ and deleted in $O((1/\sqrt{\eps})\log^3n)$ amortized time.
\end{theorem}

\section{Non-uniform power}
\label{sec:non-uniform}
Let $q$ be a receiver.
For a transmitter $s \in S$, the \emph{strength} of its signal at $q$ is $\nrg(s,q)=\frac{p(s)}{|qs|^\alpha}$ and the (multiplicatively-weighted) distance between $q$ and $s$ is $\dist(q,s) = \nrg(s,q)^{-1/\alpha} = \frac{1}{p(s)^{1/\alpha}} \cdot |qs|$. Let $s$ be the closest transmitter to $q$ according to $\dist$. Set $\intrf(q)=\sum_{s' \in S\setminus\{s\}} \nrg(s',q)$, then $\sinr(q,s)=\frac{\nrg(s,q)}{\intrf(q)}$, where we once again assume for clarity of presentation that there is no background noise, i.e., $N=0$. When $s$ is the closest transmitter to $q$, we will write $\sinr(q)$ instead of $\sinr(q,s)$.

Fix $\eps > 0$. Again, we wish to approximate $\sinr(q)$ by computing
$\tintrf(q)$ such that $\intrf(q) \leq \tintrf(q) \leq (1+\eps)\intrf(q)$ and setting $\tsinr(q)=\frac{\nrg(s,q)}{\tintrf(q)}$.
%
%
As in the uniform case, we start with a more straightforward but less efficient solution and then improve it.

\subsection{Conical shells}
\label{sec:conic_shells}
Let $q$ be a receiver and
let $s \in S$ be the closest transmitter to $q$ according to $\dist$, i.e., the one whose signal strength at $q$ is the highest. Let $s_1,\ldots,s_{n-1}$ be the transmitters in $S \setminus \{s\}$, and assume without loss of generality that $s_1$ is the second closest transmitter to $q$ among the transmitters in $S$. Recall that $\intrf(q) = \sum_{i=1}^{n-1} \nrg(s_i,q)$ and that we wish to compute a value $\tintrf(q)$ such that $\intrf(q) \le \tintrf(q) \le (1+\eps)\intrf(q)$.
We will need the following simple observation.
\begin{observation}
\label{obs:intrf_non-uni}
$\intrf(q)$ is the sum of $n-1$ positive terms of which $\nrg(s_1,q)$ is the largest, so we have
$\nrg(s_1,q) \le \intrf(q) \le (n-1)\nrg(s_1,q) \le n \cdot \nrg(s_1,q)$.
\end{observation}

\subparagraph*{Query algorithm}

Let $q$ be a query point. First, we find $s$ and $s_1$, as defined above.
Next, we set $e_0 = \frac{\eps}{2n} \nrg(s_1,q)$, and divide the transmitters in $S \setminus \{s\}$ into two subsets, $S_c$ and $S_f$, where $S_c$ consists of the transmitters with signal strength at $q$ greater than $e_0$, and $S_f$ of the remaining ones. %
We now approximate the overall contribution to $\intrf$ of the transmitters in $S_f$ and in $S_c$ separately and let $\tintrf$ be the sum of the two approximations.

The contribution of a single transmitter $s_i \in S_f$ to the sum $\intrf(q)$ is
$\nrg(s_i,q) \le e_0 = \frac{\eps}{2n} \nrg(s_1, q)$,
for a total of at most $|S_f| \cdot \frac{\eps}{2n} \nrg(s_1,q) \le \frac{\eps}{2} \nrg(s_1,q)$ over all of $S_f$.

\begin{figure}[h]
    \centering
        \includegraphics[width=0.5\textwidth]{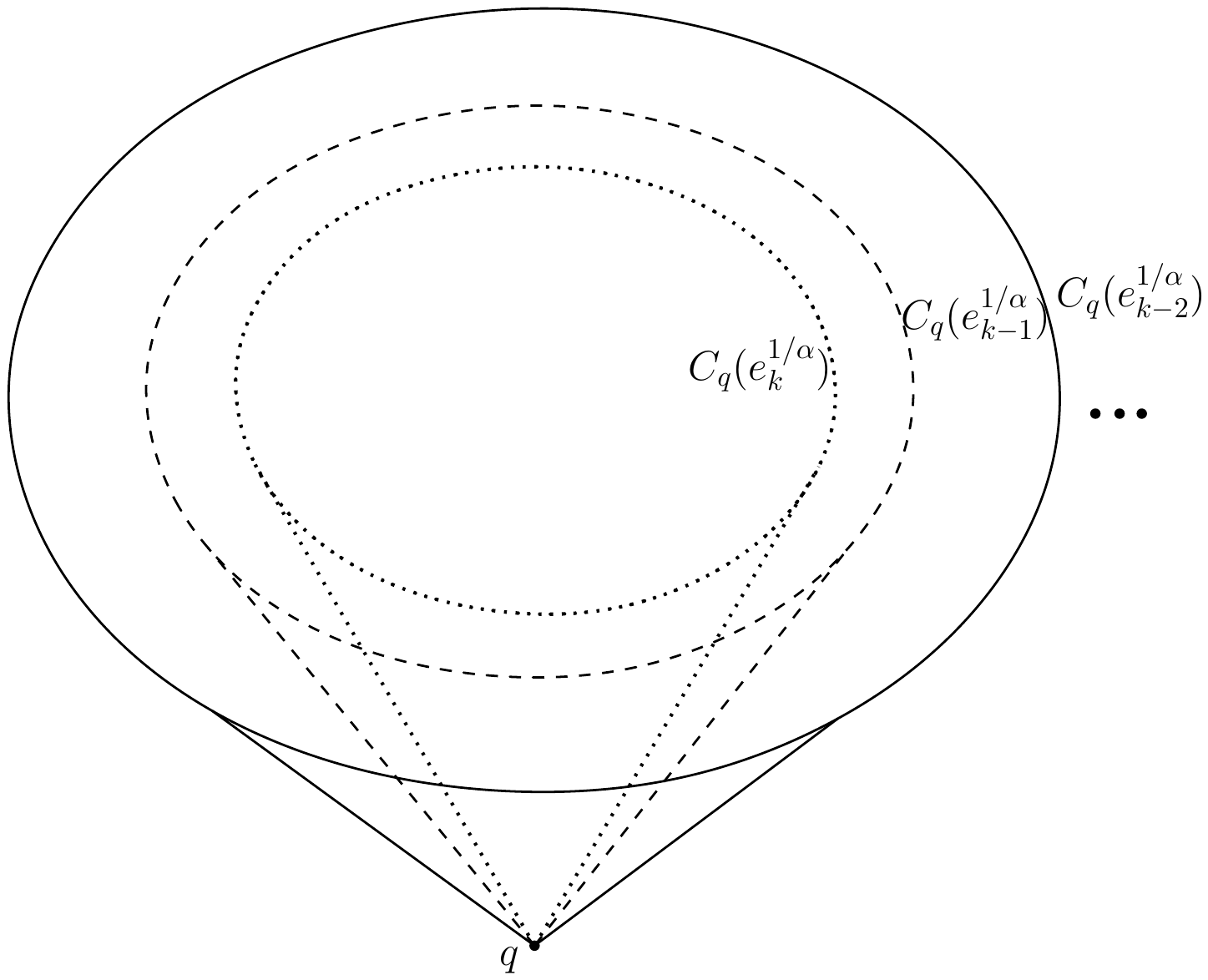}.
    \caption{Partitioning $D_q(e_0^{1/\alpha},e_k^{1/\alpha})$ into sub-shells.}
    \label{shell_partition}
\end{figure}

We identify the plane containing the transmitters and receivers with the $xy$-plane in $\mathbb{R}^3$.
Let $C_q(\rho)$ denote (the surface of) the vertical cone with apex $q$ whose $z$-coordinate at $t=(t_x,t_y)$ is $\rho|qt|$, where $\rho > 0$ is a constant. Let $D_q(\rho_1,\rho_2)$, $\rho_2 > \rho_1 > 0$,  be the set of all points in 3-space lying above (i.e., in the interior of) the cone $C_q(\rho_1)$ and below or on (i.e., not in the interior of) the cone $C_q(\rho_2)$. Informally, $D_q(\rho_1,\rho_2)$ is the region between $C_q(\rho_1)$ and $C_q(\rho_2)$; we call it a \emph{(conical) shell}.

Recall that $e_0=\frac{\eps}{2n}\nrg(s_1,q)$. Let $e_i = (1 + \frac{\eps}{2})e_{i-1}$, for $i=1,\ldots,k-1$, where $k-1$ is the largest integer for which $e_i < \nrg(s_1,q)$, and set $e_k = \nrg(s_1,q)$.
We partition the range $I=(e_0,e_k]$ of signal strengths at $q$ into $k$ sub-ranges, $I_1=(e_0,e_1], I_2=(e_1,e_2],\ldots,I_k=(e_{k-1},e_k]$, and count, for each sub-range $I_j$, the number of transmitters whose signal strength at $q$ lies in $I_j$.

Consider a sub-range $I_j = (e_{j-1},e_j]$; we want to count the number transmitters whose signal strength at $q$ lies in $I_j$.  This occurs whenever
$e_{j-1}^{1/\alpha}|qs_i| < p_i^{1/\alpha} \le e_j^{1/\alpha}|qs_i|$,
or whenever the point $(s_i, p_i^{1/\alpha})$ in $\mathbb{R}^3$ lies in the shell $D_q(e_{j-1}^{1/\alpha},e_j^{1/\alpha})$. Thus, we have reduced the problem to the difference of two conical range-counting queries.

We raise each of the transmitters $s_i \in S \setminus \{s\}$ to height $p_i^{1/\alpha}$, and preprocess the resulting set of points for conical range counting queries. If the number of points in the shell corresponding to $I_j$ is $x_j$, then we add the term $e_j x_j$ to our approximation of $\intrf(q)$, that is, we approximate the contribution of each transmitter $s_i$ whose corresponding point lies in the shell by $e_j$. (This corresponds to  vertically projecting the point $(s_i,p_i^{1/\alpha})$ onto the cone $C_q(e_j^{1/\alpha})$.)
We prove below that this yields a $(1 + \frac{\eps}{2})$-approximation of the overall contribution of the transmitters in $S_c$ to $\intrf(q)$.

\begin{lemma}
\label{lem:shell_moving_s_i}
Let $s_i \in S_c$ and let $D=D_q(e_{j-1}^{1/\alpha},e_j^{1/\alpha})$ be the shell containing $s_i$.
Then, by replacing $s_i$'s contribution by $e_j$, one obtains a $(1+\frac{\eps}{2})$-approximation of the contribution of~$s_i$ to $\intrf(q)$.
\end {lemma}
\begin{proof}
Since $s_i \in D$,
$e_{j-1}^{1/\alpha}|qs_i| < p_i^{1/\alpha} \le e_j^{1/\alpha}|qs_i|$. Moreover, by construction, $e_j/e_{j-1} \le 1+\frac{\eps}{2}$. So, the ratio between the calculated contribution of $s_i$ and its real contribution is
$\frac{e_j}{\nrg(s_i,q)} < \frac{e_j}{e_{j-1}} \le 1+\frac{\eps}{2}$.
\end{proof}
\begin{corollary}
\label{cor:approx_S_c_contribution2}
By doing this for each transmitter in $S_c$, one obtains a $(1 + \frac{\eps}{2})$-approximation of the overall contribution of the transmitters in $S_c$ to $\intrf(q)$.
\end{corollary}

It remains to show that $\tintrf(q)$, which is the sum of the approximations for $S_f$ and for $S_c$, satisfies the requirements, i.e., that $\intrf(q) \le \tintrf(q) \le (1 + \eps)\intrf(q)$.
From the description above it is clear that $\tintrf(q) \ge \intrf(q)$, so we only need to establish the upper bound. Indeed, using \cref{obs:intrf_non-uni}, we conclude that
\begin{align*}
  \tintrf(q)
  & \le \frac{\eps}{2|qs_1|^\alpha} + (1 + \frac{\eps}{2}) \sum_{s_i \in S_C} \nrg(s_i,q) \\
  & \le \frac{\eps}{2}\intrf(q) + (1 + \frac{\eps}{2})\intrf(q) = (1 + \eps)\intrf(q)\,.
\end{align*}

\subparagraph*{Implementation}
A straightforward calculation shows (analogously to \cref{lem:number_annuli}) that
$k$, the number of shells into which  $D_q(e_0^{1/\alpha},e_k^{1/\alpha})$ is partitioned, is
$O(\frac{1}{\eps}\log n)$.

We preprocess the set of raised transmitters for cone range reporting/counting queries. Then, given a query point $q$, we find $s$ and $s_1$ as follows.
Pick a random sample $T$ of $\sqrt{n}\log n$ transmitters and let $t_1 \in T$ be the transmitter whose signal strength at $q$ is the strongest. With high probability, the number of transmitters in $S$ that are closer to $q$ than $t_1$ in terms of signal strength at $q$ is $O(\sqrt{n})$, and we perform a range reporting query with the cone $C_1$ corresponding to $t_1$ in order to find them. The closest and second-closest points among the reported points are clearly $s$ and $s_1$.

As for shell range counting queries, for each such query we issue two cone range counting queries --- with the outer cone and the inner cone --- and return the difference of the answers.

We omit the time and space analysis of this version, since we describe below a more efficient variant, in which cones are replaced by pyramids.

\subsection{Pyramidal shells}
\label{sec:pyramidal_shells}

We now replace the conical shells by pyramidal ones to obtain an improved solution. Set $x=(1+\frac{\eps}{2})^\frac{1}{\alpha}$, and consider any three cones $C_q(\rho_0)$, $C_q(\rho_1)$ and $C_q(\rho_2)$ with apex at $q$, such that $\rho_0/\rho_1 = \rho_1/\rho_2 = \sqrt x$.
Let $P_q(\rho_i)$ be a regular $l$-pyramid inscribed in $C_q(\rho_i)$, where $l = \bigl\lceil\frac{\pi}{\sqrt{2- \frac{2}{\sqrt{x}}}}\bigr\rceil$.  That is, $P_q(\rho_i)$'s apex is at $q$, its edges emanating from $q$ are contained in (the surface of) $C_q(\rho_i)$, and the cross section of $P_q(\rho_i)$ and $C_q(\rho_i)$, using any horizontal cutting plane above $q$, is a regular $l$-gon and its circumcircle, respectively. The \emph{pyramidal shell} defined by $P_q(\rho_2)$ and $P_q(\rho_1)$ and denoted $\ps_q(\rho_2,\rho_1)$ is the semi-open region consisting of all points in the interior of $P_q(\rho_2)$ but not in the interior of $P_q(\rho_1)$. From \cref{cl:poly_is_included} and the observation above, it follows that $C_q(\rho_{i-1})$ is contained in $P_q(\rho_i)$, for $i=1,2$, and therefore, $\ps_q(\rho_2,\rho_1)$ is contained in
$D_q(\rho_2,\rho_0)$.

\subparagraph*{Query algorithm}
We highlight the differences with the conical-shell based approach.
First, we find $s$ and $s_1$, the closest and the second-closest transmitters to $q$, respectively, as described in detail below.
Previously, the transmitters were divided into two subsets lying close to $q$ and lying far from it, with the threshold $e_0 =\frac{\eps}{2n} \cdot \nrg(s_1,q)$. Here, we set $e'_0 = \frac{\nrg(s_1,q)}{x^{m/2}}$, where $m$ is the smallest integer for which $\frac{\nrg(s_1,q)}{x^{m/2}} \le \frac{e_0}{\sqrt x}$, and consider a transmitter close to $q$ whenever it lies in the interior of the pyramid $P_q({e'_0}^\frac{1}{\alpha})$, i.e., the pyramid inscribed in $C_q({e'_0}^\frac{1}{\alpha})$. The contribution of a single transmitter $s_i \in S_f$ to the sum $\intrf(q)$ is
$\nrg(s_i,q) \le \sqrt{x}e'_0 \le e_0  = \frac{\eps}{2n} \cdot \nrg(s_1,q)$,
for a total of at most $\frac{\eps}{2} \cdot \nrg(s_1,q)$, as before.

Consider the conical shell $D_q({e'_0}^\frac{1}{\alpha},\nrg(s_1,q)^\frac{1}{\alpha})$ and partition it into $m$ conical shells, such that the ratio between the parameters of the inner and outer cone of a shell is $\sqrt{x}$.
For each of the $m+1$ cones defining these conical shells, draw its inscribed regular $l$-pyramid. Let $P_q(\rho_1),...,P_q(\rho_{m+1})$ be the resulting sequence of pyramids, where $P_q(\rho_1)$ is the innermost one, and consider the corresponding sequence of $m$~nested pyramidal shells. Notice that $m=O(k)$, where $k$ is the number of cones in the conical shells version,
so once again
$m = O(\frac{1}{\eps}\log n)$. Moreover, each of the pyramidal shells, except for $\ps_q(\rho_2,\rho_1)$, is contained in the union of two consecutive conical shells, which is a conical shell of ratio $x$. For $\ps_q(\rho_2,\rho_1)$, we observe that $S_c \cap \ps_q(\rho_2,\rho_1)$ is contained in the innermost conical shell.

We assign a transmitter $s_i$ to a shell $\ps_q(\rho_{j+1},\rho_j)$ if $s_i$ lands in the shell after being raised to height $p_i^{1/\alpha}$.
Now, for each shell $\ps_q(\rho_{j+1},\rho_j)$ and each $s_i \in S_c$ assigned to it, we estimate the contribution of $s_i$ from above by $1/\rho_{j-1}^\alpha$, i.e., by projecting $s_i$ onto the inner cone of
the conical shell %
containing $\ps_q(\rho_{j+1},\rho_j)$. By \cref{lem:shell_moving_s_i} and the subsequent corollary, we obtain a $(1 + \frac{\eps}{2})$-approximation of the overall contribution of the transmitters in $S_c$ to $\intrf(q)$.  Adding our previous estimate for those in $S_f$ yields the promised $\tintrf(q)$.

\subparagraph*{Implementation}

Observe that each regular $l$-pyramid $P_q(\rho_j)$ is the union of $l$ 3-sided wedges, where the $i$th wedge is defined by two planes of fixed orientation (perpendicular to the $xy$-plane) and a third plane containing the $i$th face of the pyramid.

In the preprocessing stage we construct $l$ data structures over the set $S$, one for each family of wedges, supporting dynamic 3-dimensional 3-sided wedge range counting queries (a restricted form of simplex range counting in three dimensions).
Each data structure handles wedges of the same ``type''; the orientations of the two vertical bounding planes are fixed, while the orientation of the third plane varies (but remains perpendicular to the vertical plane bisecting the first two).
The data structure for the $i$th family is a three-level search structure, where the first two levels allow us to represent the points of $S$ that
lie in the 2-sided wedge formed by the two vertical planes delimiting our 3-sided wedge, as a small collection of canonical subsets. For each canonical subset of the second level of the structure, we raise each of its points $s_i$ to height $p_i^{1/\alpha}$ and then project it onto a vertical plane which is parallel to the bisector of the two vertical wedge boundaries. Finally, we construct for the resulting set of points a data structure for two-dimensional halfplane range counting queries.
We will also need the corresponding reporting structure, see below.

Using standard tools for dynamic multilevel structures and, for example, Matou\v{s}ek's data structure for halfplane range counting at the bottom level, we obtain a structure of $O(n\polylog n)$ size that supports wedge counting (and reporting) queries in $O(n^{1/2}\polylog n)$ time and updates in $O(\polylog n)$ amortized time.

Now, given a pyramidal shell $\ps_q(\rho_{j+1},\rho_j)$, we can count the number of raised points that lie in it as follows. We first perform $l$ queries for the pyramid $P_q(\rho_{j+1})$, one in each of the $l$ data structures, to obtain the total number of points that lie in it.  We repeat the process for the pyramid $P_q(\rho_j)$ and finally subtract the latter number from the former one.

Below
we describe how to find $s$ and $s_1$, the closest and second-closest points to~$q$, in randomized $O(\sqrt{n}\log n)$ time with high probability plus $l$ wedge reporting queries. %
Once again, an update is performed by modifying the underlying data structures.
We summarize the main result of this section.

\begin{theorem}
\label{th:non-uni-dyn}
	One can preprocess $n$ arbitrary-power transmitters, in time and space $O(\frac{n}{\sqrt{\eps}} \polylog n)$, into a dynamic data structure for approximate SINR queries. A query is handled in $O(\frac{\sqrt n}{\eps^{3/2}}\polylog n)$ time, after a preliminary stage which is performed in randomized time $O(\frac{\sqrt n}{\sqrt{\eps}}\polylog n)$ with high probability, while an update is performed in amortized time $O(\frac{1}{\sqrt{\eps}}\polylog n)$.
\end{theorem}

\subparagraph*{Finding the closest and second-closest transmitters to $q$}
We have assumed that given a query point $q$, we can find the closest $s$ and second-closest $s_1$ transmitters to $q$, efficiently. This section deals with this initial stage.

We begin by observing that we do not really need to find $s_1$, provided that we can obtain a sufficiently good approximation of $\nrg_1 =  \nrg(s_1,q)$. Let $\tnrg_1 \in \mathbb{R}$ such that $\tnrg_1 \le \nrg_1 \le (1+\delta)\tnrg_1$, where $\delta > 0$ is a sufficiently small constant. Then, it is easy to modify our query algorithm so that it uses $\tnrg_1$ instead of $\nrg_1$. For simplicity of presentation, we refer in this paragraph to the algorithm using conical shells. Set $\tilde{e}_0 = \frac{\eps}{2n}\tnrg_1$. A transmitter will be considered close to $q$ if and only if its signal strength at $q$ is greater than~$\tilde{e}_0$. If $s_i \in S$ is far from $q$, then $\nrg(s_i,q) \le \tilde{e}_0 \le \frac{\eps}{2n}\nrg_1 = \frac{\eps}{2n}\nrg(s_1,q)$, so the overall contribution of the transmitters in $S_f$ is bounded by $\frac{\eps}{2}\nrg(s_1,q)$, as before.
Next, we partition the range $I = (\tilde{e}_0, (1+\delta)\tnrg_1]$ into $k = O(\frac{1}{\eps}\log n)$ sub-ranges, such that the ratio between the extreme values of a sub-range is at most $1 + \eps/2$ and proceed exactly as before.

We now describe how to find $s$. Our algorithm may or may not find $s_1$. However, if it does not find $s_1$, it returns a transmitter $t_1 \in S$ such that $\nrg(t_1, q) \le \nrg(s_1, q) \le (1+\delta)\nrg(t_1, q)$, where $\delta > 0$ is a sufficiently small constant, so we can set $\tnrg_1 = \nrg(t_1, q)$ and apply the above modified query algorithm.

Pick a random sample $T$ of $\sqrt{n}\log n$ transmitters and let $t_1 \in T$ be the transmitter whose signal strength at $q$ is the strongest. This can be done in $O(\sqrt{n}\log n)$ time.
With high probability the number of transmitters in $S$ that are closer to $q$ than $t_1$, in terms of signal strength at $q$, is $O(\sqrt{n})$.

We first lift each transmitter $s=(s_x, s_y) \in S$ to the point $\hat{s} = (s_x, s_y, p(s)^{1/\alpha})$.
Draw the cone $C_1$ corresponding to $t_1$, i.e., the cone whose $z$-coordinate above point $s$ is $\nrg(t_1,q)^{1/\alpha}|qs| = (p(t_1)^{1/\alpha}/|qt_1|)|qs|$.
Let $l$ be as above and consider the $l$-pyramid $P_1$ inscribed in $C_1$.
Let $C_0$ be the cone inscribed in $P_1$, so that $P_1$ lies between $C_0$ and $C_1$.
Notice that $C_0$ is the cone whose $z$-coordinate above point $s$ is $(1+\delta) \nrg(t_1,q)^{1/\alpha}|qs|$ (where we set $\delta = (1+\frac{\eps}{2})^{\frac{1}{2\alpha}} - 1)$.

Perform a range reporting query with $P_1$ (i.e., find all lifted points that lie in the interior of $P_1$ or on $P_1$). Since $P_1$ is inside $C_1$, with high probability the number of points in $P_1$ is $O(\sqrt{n})$. If the resulting set is non-empty, then in randomized $O(\sqrt{n})$ time with high probability we can find $s$ and also $s_1$ (provided the number of returned points is greater than~1).

Otherwise, if $P_1$ is empty, we claim that the answer to the SINR query must be \textsc{no}, i.e., $q$ cannot receive any transmitter. Indeed, in the best scenario $\hat{s}$ lies on $C_0$, where $s \in S$ is the closest transmitter to $q$, and the rest of the $O(\sqrt{n})$ transmitters, lifted to 3-space, lie on the cone $C_1$. But this will imply that $\sinr(q) < 1$.
Indeed $\nrg(s,q)=(1+\delta)^\alpha \nrg(t_1,q)$ and, for any other of the $\sqrt{n}$ transmitters $s'$, $\nrg(s',q)=\nrg(t_1,q)$, implying $sinr(q) < (1+\delta)^\alpha / \sqrt{n} \ll 1$.

If only one point lies in $P_1$, then we use $t_1$ as an approximation of $s_1$ as described above.

\section{Successive interference cancellation (SIC)}
\label{sec:sic}

Fix a receiver location $q$. 
SIC is a technique that enables $q$ to receive a specific transmitter~$t$, even when $\sinr(q,t) < \beta$.  More specifically, order the transmitters $s_1, \ldots, s_n$ in $S$ by increasing signal strength at $q$, assume $t=s_k$, and let $\sinr_i(q)$ denote the SIN ratio for the signal of $s_i$ at $q$, while ignoring transmitters $s_1,\dots,s_{i-1}$.  If $\sinr_1(q)=\sinr(q,s_1) \ge \beta$, $q$~can subtract $s_1$'s signal from the combined signal. If, in addition, $\sinr_2(q) \ge \beta$, $q$ can also subtract $s_2$ from the combined signal of the transmitters $s_2,\dots,s_n$, and so on.  If $\sinr_i(q)\geq \beta$, for $i=1,\dots,k$, we say that \emph{SIC succeeds for~$s_k$ at~$q$, in~$k$ rounds}.
We can simulate this process using our data structures for approximate SINR queries via a sequence of $k$ queries and $k-1$ deletions and insertions, and determine (approximately) whether SIC succeeds for~$s_k$ at~$q$. 
Observe that we need $t$ only to terminate the query, while Avin et al.~\cite{AvinCHKLPP17} need $t$ to identify the part of the data structure in which to initiate the search; in particular, we can generate all the transmitters accessible via SIC given a location $q$ in polylogarithmic time per transmitter, while they need to consult each of the $n$ parts of the data structure.
We obtain the following theorem.
\begin{theorem}
Assuming $t=s_k$, the 
simulation above can be performed 
in amortized time $O((1/\eps^{3/2})k\polylog n)$ in the uniform-power version. In the non-uniform version, it can be performed in
$O((1/{\eps^{3/2}})k\sqrt{n}\polylog n)$ time (see \cref{th:non-uni-dyn} for details).
\end{theorem}

\begin{remark}
Consider, e.g., the uniform power version. Given $q$, $s_1$ can be canceled only if $\sinr_1(q) \ge \beta$, which implies that $|qs_2| > |qs_1|\beta^{1/\alpha}$. Now, $s_2$ can be canceled only if $s_1$ can be canceled and $\sinr_2(q) \ge \beta$, which implies that $|qs_3| > |qs_2|\beta^{1/\alpha} > |qs_1|\beta^{2/\alpha}$, etc. Therefore, in practice, it is unlikely that $s_k$ can be canceled for, say, $k > c\alpha \log_\beta n$, where $c \ge 1$ is some constant, since this would imply that $|qs_k| > |qs_1|\beta^{(c\alpha\log_\beta n)/\alpha} = |qs_1|n^c$. So, in practice the number of iterations needed to determine whether $t$ can be received or not using SIC will not exceed $c\log_\beta n$, and therefore the entire simulation can be performed in $O((1/\eps^{3/2})\polylog n)$ amortized time. A similar argument shows that in the non-uniform power version, the entire simulation can be performed in practice in $O((1/{\eps^{3/2}})\sqrt{n}\polylog n)$ time (see \cref{th:non-uni-dyn} for details), assuming the ratio of highest to lowest power is bounded by a polynomial in $n$.
\end{remark}

\section{Resolving SINR queries --- Back to the static setting}
\label{sec:ultimate}

In this section, we present a solution for the non-uniform power version in a static setting, which enables one to approximately answer an SINR query in $O(\polylog n)$ time, after near-linear time preprocessing.

We will use the following theorem of Har-Peled and Kumar~\cite[Theorem~2.16]{wann}:
\begin{theorem}[Approximate Two-Dimensional Multiplicatively Weighted Nearest Neighbor \cite{wann}]
  \label{fact:2d-ann-wtd}
  Given an $n$-point set $P \subset \RR^2$ with positive weights $w_p$ and a positive number $\eps$, one can preprocess it into a data structure of space $O(n\eps^{-6}\log^4 n)$ in time $O(n\eps^{-6}\log^7 n)$ to support $O(\log(n/\eps))$-time queries of the form: Given a point~$q$, return $p' \in P$ so that $|p'q|/w_{p'} \leq (1+\eps) |p^*q|/w_{p^*}$, where $p^*$ is the point in $P$ minimizing $|pq|/w_p$.
\end{theorem}

Let $\eps > 0$. Let $q$ be a query point, let $s$ be the weighted-closest transmitter to $q$ and let $s_1$ be the second weighted-closest transmitter to $q$.

Let $\eps_1 = \eps_1(\eps,\alpha) > 0$, to be fixed below.
We construct the data structure of Har-Peled and Kumar for approximate weighted nearest neighbor queries in $S$ with error $\eps_1$, where the weight of transmitter $s \in S$ is $1/p^{1/\alpha}$, and store it in the root $r$ of a binary tree $T$. Next, we (arbitrarily) divide $S$ into two subsets $S_{\text{\sc l}}$ of size $\lfloor n/2 \rfloor$ and $S_{\text{\sc r}}$ of size $\lceil n/2 \rceil$, and construct the data structure of Har-Peled and Kumar for each of these subsets. We store the structures for $S_{\text{\sc l}}$ ad $S_{\text{\sc r}}$ in the left and right children of $r$, respectively. Finally, we apply this last step recursively to the left child and to the right child of $r$.

Let $s^q$ be the transmitter returned when performing a query with $q$ in the structure stored in the root $r$.
Let $s_1^q$ be the transmitter returned when performing $O(\log n)$ queries with $q$ in the structures stored in the nodes hanging from the path in $T$ from $s^q$ to the root $r$. (Among the $O(\log n)$ candidates we choose the one whose signal strength at $q$ is the maximum.)

If $\nrg(s^q,q) < \nrg(s_1^q,q)$, then interchange $s^q$ and $s_1^q$.

We know that 
$|sq|/p(s)^{1/\alpha} \le |s^qq|/p(s^q)^{1/\alpha} \le (1+\eps_1)|sq|/p(s)^{1/\alpha}$, or 
\[
  \frac{1}{(1+\eps_1)^\alpha}\nrg(s,q) \le \nrg(s^q,q) \le \nrg(s,q)\,,
\]
and that $|s_1q|/p(s_1)^{1/\alpha} \le |s_1^qq|/p(s_1^q)^{1/\alpha} \le (1+\eps_1)|s_1q|/p(s_1)^{1/\alpha}$, or
\[
  \frac{1}{(1+\eps_1)^\alpha}\nrg(s_1,q) \le \nrg(s_1^q,q) \le \nrg(s_1,q)\,.
\]
Set $\tnrg_1 = \nrg(s_1^q,q)$ and choose $\eps_1=\Theta(\eps/\alpha)$, so that $(1+\eps_1)^\alpha=1+\eps$, and therefore $\nrg(s,q) \le (1+\eps)\nrg(s^q,q)$ and $\nrg(s_1,q) \le (1+\eps)\nrg(s_1^q,q)$.
As before, set $x=(1+\frac{\eps}{2})^\frac{1}{\alpha}$.
Set $e_0 = \frac{\eps}{2n}\tnrg_1$ and set $e'_0 = \frac{\tnrg_1}{x^{m/2}}$, where $m$ is the smallest integer for which $\frac{\tnrg_1}{x^{m/2}} \le  \frac{e_0}{\sqrt x}$, and consider a transmitter close to $q$ whenever it lies in the interior of the pyramid $P_q({e'_0}^\frac{1}{\alpha})$, i.e., the pyramid inscribed in $C_q({e'_0}^\frac{1}{\alpha})$. This defines the sets $S_c$ and $S_f$ of transmitters close and far from $q$.

The contribution of a single transmitter $s_i \in S_f$ to the sum $\intrf(q)$ is
$\nrg(s_i,q) \le \sqrt{x} e'_0 \le e_0  = \frac{\eps}{2n} \cdot \tnrg_1$,
for a total of at most $\frac{\eps}{2} \cdot \tnrg_1 \le \frac{\eps}{2} \cdot \nrg(s_1,q)$.

Consider the conical shell $D_q({e'_0}^\frac{1}{\alpha},(1+\eps_1)\tnrg_1^\frac{1}{\alpha})$ and partition it into $m$ conical shells, such that the ratio between the parameters of the inner and outer cone of a shell is $\sqrt{x}$.
For each of the $m+1$ cones defining these conical shells, draw its inscribed regular $l$-pyramid. Let $P_q(\rho_1),...,P_q(\rho_{m+1})$ be the resulting sequence of pyramids, where $P_q(\rho_1)$ is the innermost one, and consider the corresponding sequence of $m$~nested pyramidal shells. Notice that $m=O(k)$, where $k$ is the number of cones in the conical shells version,
so once again
$m = O(\frac{1}{\eps}\log n)$. Moreover, each of the pyramidal shells, except for $\ps_q(\rho_2,\rho_1)$, is contained in the union of two consecutive conical shells, which is a conical shell of ratio $x$. For $\ps_q(\rho_2,\rho_1)$, $S_c \cap \ps_q(\rho_2,\rho_1)$ is contained in the innermost conical shell.

We proceed as in the dynamic setting. We construct $l$ three-level data structures. However, at the bottom level of each of these structures, we store data structures for 2-dimensional \emph{approximate} half-plane range counting (instead of exact counting).
(Refer to \cite{AfshaniC09} for a randomized data structure of expected size $n/\epsilon^{O(1)}$ constructed in expected time $(n/\epsilon^{O(1)})\log n$ for answering halfplane range-countng queries $(1+\epsilon)$-approximately in expected time $\epsilon^{-O(1)}\log n$, for any $\epsilon>0$.)

Each pyramid $P_q(\rho_i)$ is the union of $l$ pyramidal wedges.
Consider any one of the $l$ families of pyramidal wedges, $W_1,...,W_{m+1}$, where $W_1$ is the innermost wedge and $W_{m+1}$ is the outermost one. We need to count the number of points in each of the shells $W_i - W_{i-1}$, for $i = 2,\ldots,m+1$. Let $c_i$ be the number of points in wedge $W_i$ and let $c'_i$ be the number of points returned by the data structure when querying with $W_i$. Then $(1-\delta)c_i \le c'_i \le (1+\delta)c_i$, for some $\delta = \delta(\eps)$.

We would have liked to compute the sum $X = \sum_{i=2}^{m+1} e_{i-1}(c_i - c_{i-1})$, as the contribution of this family of shells to the approximated interference. However, this would take roughly $\sqrt{n}$ time. Instead, we observe that
\begin{align*}
X & = \sum_{i=2}^{m+1} e_{i-1}(c_i - c_{i-1}) = \sum_{i=2}^{m+1} e_{i-1}c_i - \sum_{i=2}^{m+1} e_{i-1}c_{i-1} \\
  & = \sum_{i=2}^{m+1} e_{i-1}c_i - \sum_{i=1}^m e_ic_i = e_mc_{m+1} + \sum_{i=2}^m (e_{i-1} - e_i) c_i - e_1c_1 \\
	& = e_mc_{m+1} + \sum_{i=2}^m (e_{i-1} - e_i) c_i \,,
\end{align*}
where the last equality follows from the fact that $c_1 = 0$. The above series of equations shows that, since $X$ is a linear combination of $c_2, \ldots, c_{m+1}$ with positive coefficients, replacing $c_2, \ldots, c_{m+1}$ by their relative approximations $c'_2, \ldots, c'_{m+1}$ will yield a relative approximation of $X$. We thus obtain a value $X' = e_mc'_{m+1} + (\sum_{i=2}^m (e_{i-1} - e_i) c'_i)$ that satisfies $(1 - \delta)X \le X' \le (1+\delta)X$.

Finally, set $\delta=\Theta(\eps)$ and adjust the approximation parameters by a constant factor so that the combined multiplicative error in both the numerator and denominator for the expression for $\sinr$ does not exceed $1+\eps$.  We thus conclude:

\begin{theorem}
  One can preprocess $n$ arbitrary-power transmitters, in expected $O(\frac{n}{\eps^{O(1)}} \polylog n)$ time and space, into a data structure that can answer approximate SINR queries in expected $O(\frac{1}{\eps^{O(1)}} \polylog n)$ time.
\end{theorem}

\subparagraph*{Acknowledgments}
We wish to thank Pankaj K. Agarwal, Timothy Chan, Sariel Har-Peled, and Wolfgung Mulzer for discussions, hints, and outright help with some aspects of this paper.
\bibliographystyle{siamplain}
\bibliography{sinr}

\begin{thebibliography}{10}

\bibitem{AfshaniC09}
{\sc P.~Afshani and T.~M. Chan}, {\em On approximate range counting and depth},
  Discrete {\&} Computational Geometry, 42 (2009), pp.~3--21,
  \url{https://doi.org/10.1007/s00454-009-9177-z}.

\bibitem{a-rs-04}
{\sc P.~K. Agarwal}, {\em Range searching}, in Handbook of Discrete and
  Computational Geometry, Second Edition, CRC Press LLC, 2004, pp.~809--837,
  \url{https://doi.org/10.1201/9781420035315.ch36}.

\bibitem{AronovBK18}
{\sc B.~Aronov, G.~Bar-On, and M.~J. Katz}, {\em Resolving {SINR} queries in a
  dynamic setting}, in Proceedings of Automata, Languages, and Programming ---
  45th International Colloquium, {ICALP} 2018, Part {III}, 2018,
  pp.~145:1--145:13.

\bibitem{AK-TAlg}
{\sc B.~Aronov and M.~J. Katz}, {\em Batched point location in {SINR} diagrams
  via algebraic tools}, ACM Transactions on Algorithms, 14 (2018),
  pp.~41:1--41:29.

\bibitem{AvinCHKLPP17}
{\sc C.~Avin, A.~Cohen, Y.~Haddad, E.~Kantor, Z.~Lotker, M.~Parter, and
  D.~Peleg}, {\em {SINR} diagram with interference cancellation}, Ad Hoc
  Networks, 54 (2017), pp.~1--16,
  \url{https://doi.org/10.1016/j.adhoc.2016.08.003}.

\bibitem{aeklpr-sdciawn-12}
{\sc C.~Avin, Y.~Emek, E.~Kantor, Z.~Lotker, D.~Peleg, and L.~Roditty}, {\em
  {SINR} diagrams: Convexity and its applications in wireless networks}, J.
  ACM, 59 (2012), pp.~18:1--18:34,
  \url{https://doi.org/10.1145/2339123.2339125}.

\bibitem{Chan19}
{\sc T.~M. Chan}, {\em Dynamic geometric data structures via shallow cuttings},
  in 35th International Symposium on Computational Geometry, SoCG 2019, June
  18-21, 2019, Portland, Oregon, {USA.}, 2019, pp.~24:1--24:13,
  \url{https://doi.org/10.4230/LIPIcs.SoCG.2019.24}.

\bibitem{ChanPC}
{\sc T.~M. Chan}.
\newblock Personal communication, September 2019.

\bibitem{c-fadsi-88}
{\sc B.~Chazelle}, {\em A functional approach to data structures and its use in
  multidimensional searching}, SIAM J. Comput., 17 (1988), pp.~427--462.

\bibitem{bcko-cgaa-08}
{\sc M.~de~Berg, O.~Cheong, M.~van Kreveld, and M.~H. Overmars}, {\em
  Computational Geometry: Algorithms and Applications}, Springer-Verlag,
  Berlin, 3rd~ed., 2008, \url{http://www.cs.ruu.nl/geobook/}.

\bibitem{wann}
{\sc S.~Har-Peled and N.~Kumar}, {\em Approximating minimization diagrams and
  generalized proximity search}, SIAM J. Comput., 44 (2015), pp.~944--974.

\bibitem{he-munro-14}
{\sc M.~He and J.~I. Munro}, {\em Space efficient data structures for dynamic
  orthogonal range counting}, Comput. Geom., 47 (2014), pp.~268--281,
  \url{https://doi.org/10.1016/j.comgeo.2013.08.007}.

\bibitem{klpp-twn-11}
{\sc E.~Kantor, Z.~Lotker, M.~Parter, and D.~Peleg}, {\em The topology of
  wireless communication}, in Proceedings 43rd {ACM} Symposium on Theory of
  Computing, {STOC} 2011, 2011, pp.~383--392,
  \url{http://doi.acm.org/10.1145/1993636.1993688}.

\bibitem{KantorLPP15}
{\sc E.~Kantor, Z.~Lotker, M.~Parter, and D.~Peleg}, {\em Nonuniform
  {SINR}+{V}oroni diagrams are effectively uniform}, in Proceedings 29th
  International Symposium on Distributed Computing, {DISC} 2015, 2015,
  pp.~588--601, \url{https://doi.org/10.1007/978-3-662-48653-5_39}.

\bibitem{DBLP:conf/soda/KaplanMRSS17plus}
{\sc H.~Kaplan, W.~Mulzer, L.~Roditty, P.~Seiferth, and M.~Sharir}, {\em
  Dynamic planar {V}oronoi diagrams for general distance functions and their
  algorithmic applications}, in Proceedings 28th {ACM-SIAM} Symposium on
  Discrete Algorithms, {SODA} 2017, 2017, pp.~2495--2504,
  \url{https://doi.org/10.1137/1.9781611974782.165}.
\newblock See also arXiv:1604.03654.

\bibitem{m-ept-92}
{\sc J.~Matou{\v s}ek}, {\em Efficient partition trees}, Discrete Comput.
  Geom., 8 (1992), pp.~315--334.

\bibitem{MulzerPC}
{\sc W.~Mulzer}.
\newblock Personal communication, September 2019.

\bibitem{o-ddds-83}
{\sc M.~H. Overmars}, {\em The Design of Dynamic Data Structures}, vol.~156 of
  Lecture Notes in Computer Science, Springer-Verlag, Heidelberg, West Germany,
  1983.

\bibitem{wl-arrcd-85}
{\sc D.~E. Willard and G.~S. Lueker}, {\em Adding range restriction capability
  to dynamic data structures}, J.~Assoc. Comput. Mach., 32 (1985),
  pp.~597--617.

\end{thebibliography}

\end{document}